\documentclass[10pt]{article}

\usepackage[all]{xy}
\usepackage{amsmath}
\usepackage{amsfonts}
\usepackage{amsthm}

\newtheorem{theorem}{Theorem}
\newtheorem{lemma}{Lemma}

\newtheorem{remark}{Remark}

\newcommand{\mR}{\mathbb{R}}

\newcommand{\mE}{\mathbb{E}}
\newcommand{\mZ}{\mathbb{Z}}

\newcommand{\cD}{\mathcal{D}}

\newcommand{\cP}{\mathcal{P}}

\newcommand{\ux}{\underline{x}}
\newcommand{\uxb}{\underline{x \grave{}}}

\newcommand{\pj}{\partial_{x_j}}
\newcommand{\pjb}{\partial_{{x \grave{}}_{j}}}
\newcommand{\pkb}{\partial_{{x \grave{}}_{k}}}

\newcommand{\pI}{\partial_{x_i}}
\newcommand{\pIb}{\partial_{{x \grave{}}_{i}}}
\newcommand{\pIcont}{\partial_{x_i} \rfloor}

\newcommand{\pk}{\partial_{x_k}}
\newcommand{\px}{\partial_x}
\newcommand{\upx}{\partial_{\underline{x}}}
\newcommand{\upxb}{\partial_{\underline{{x \grave{}}} }}

\newcommand{\pxcont}{\partial_x \rfloor}
\newcommand{\pjcont}{\partial_{x_j} \rfloor}
\newcommand{\pjbcont}{\partial_{{x \grave{}}_{j}} \rfloor}
\newcommand{\pkcont}{\partial_{x_k} \rfloor}
\newcommand{\pkbcont}{\partial_{{x \grave{}}_{k}} \rfloor}

\newcommand{\fermdirac}{{e \grave{}}_{2j} \partial_{{x \grave{}}_{2j-1}} - {e \grave{}}_{2j-1} \partial_{{x \grave{}}_{2j}}}

\begin{document}

\title{A Clifford analysis approach to superspace}

\author{H.\ De Bie\thanks{Research assistant supported by the Fund for Scientific Research Flanders (F.W.O.-Vlaanderen), E-mail: {\tt Hendrik.DeBie@UGent.be}} \and F.\ Sommen\thanks{E-mail: {\tt fs@cage.ugent.be}}
}

\date{\small{Clifford Research Group -- Department of Mathematical Analysis}\\
\small{Faculty of Engineering -- Ghent University\\ Galglaan 2, 9000 Gent,
Belgium}}

\maketitle

\begin{abstract}
A new framework for studying superspace is given, based on methods from Clifford analysis. This leads to the introduction of both orthogonal and symplectic Clifford algebra generators, allowing for an easy and canonical introduction of a super-Dirac operator, a super-Laplace operator and the like. This framework is then used to define a super-Hodge coderivative, which, together with the exterior derivative, factorizes the Laplace operator. Finally both the cohomology of the exterior derivative and the homology of the Hodge operator on the level of polynomial-valued super-differential forms are studied. This leads to some interesting graphical representations and provides a better insight in the definition of the Berezin-integral.
\end{abstract}

\section{Introduction}
\label{intro}

The aim of this paper is to examine the usefulness of a novel approach to the study of superspace. This approach is inspired by ideas of Clifford analysis (see \cite{MR697564}, \cite{MR1130821} and \cite{MR1169463} and references therein), a function theory which is obtained by formally taking the square root of the Laplace operator in Euclidean space.

It is possible to reformulate the basic features of Clifford analysis in an abstract framework, called radial algebra (see \cite{MR1472163}). This radial algebra has many interesting applications, for example in the algebraic analysis of several Dirac operators and associated systems (see \cite{MR2089988} and references therein). One can also use this framework to define a superspace version of Clifford analysis (see \cite{DBS1} and an earlier attempt in \cite{MR1771370}).

One of the main features of this new approach to superspace is the introduction of the so-called super-dimension $M$, a number which can take both positive and negative integer values and acts as a global dimension. This will be illustrated in the sequel.

The new approach also offers several advantages, e.g. a lot of interesting operators are obtained in a rather easy and canonical way, such as the super-exterior derivative, the super-Hodge coderivative, the super-Laplace operator, etc. Moreover, we will use these new techniques to derive super-versions of the classical lemma of Poincar\'e and its dual. This gives a strong motivation for the definition of the Berezin integral, which is usually defined in a rather ad hoc way. Also some interesting graphical representations of spaces of polynomial-valued differential forms are obtained. Moreover, compared to the classical approaches to superspace (from either algebraic geometry, see \cite{MR732126,MR565567,MR0580292}, or from differential geometry, see \cite{MR778559,MR574696}), our approach has the advantage of requiring fewer advanced mathematical techniques, and thus being more easily accessible to the average physicist.

The paper is organized as follows. We start with an introduction to Clifford analysis on superspace. After constructing the specific algebra we need for our purposes, we define the basic differential operators and contractors. Then we introduce the basic objects typically arising in Clifford analysis, such as vector variables and differentials, the Dirac-, Euler- and Gamma-operators. Next the basic commutation and anti-commutation relations between these operators are established, which will immediately show the importance of the super-dimension. Thereafter we define a Hodge coderivative which behaves exactly as classically. In the last sections we give some diagrammatic representations, we prove super-versions of the Poincar\'e lemmata and establish a link with the definition of the Berezin integral.

\section{Clifford analysis in superspace}

\subsection{The basic objects}

We consider an algebra with six different types of generators, namely $m$ commuting co-ordinates $x_1, \ldots, x_m$ and corresponding differentials $dx_1, \ldots, dx_m$, $2n$ anti-commuting co-ordinates ${x \grave{}}_1, \ldots, {x \grave{}}_{2n}$ and corresponding differentials $d{x \grave{}}_1, \ldots, d{x \grave{}}_{2n}$, $m$ orthogonal Clifford generators $e_1, \ldots, e_m$ and $2n$ symplectic Clifford generators ${e \grave{}}_1, \ldots , {e \grave{}}_{2n}$. They have to satisfy the following relations:

\[
\begin{array}{lll} 
x_i x_j =  x_j x_i&\quad&dx_i d x_j = -d x_j  dx_i \\
{x \grave{}}_i {x \grave{}}_j =  - {x \grave{}}_j {x \grave{}}_i&\quad& dx_i d {x \grave{}}_j = -d {x \grave{}}_j dx_i\\
x_i {x \grave{}}_j =  {x \grave{}}_j x_i&\quad& d {x \grave{}}_i d {x \grave{}}_j = d {x \grave{}}_j d{x \grave{}}_i\\
\\
x_i d x_j = d x_j  x_i&\quad& e_j e_k +e_k e_j = -2 \delta_{jk}\\
x_i d {x \grave{}}_j = d {x \grave{}}_j x_i&\quad&{e \grave{}}_{2j} {e \grave{}}_{2k} - {e \grave{}}_{2k} {e \grave{}}_{2j}=0\\
{x \grave{}}_i d x_j = d x_j  {x \grave{}}_i&\quad&{e \grave{}}_{2j-1} {e \grave{}}_{2k-1} -{e \grave{}}_{2k-1} {e \grave{}}_{2j-1}=0\\
{x \grave{}}_i d {x \grave{}}_j = -d {x \grave{}}_j {x \grave{}}_i&\quad&{e \grave{}}_{2j-1} {e \grave{}}_{2k} -{e \grave{}}_{2k} {e \grave{}}_{2j-1}=\delta_{jk}\\
&\quad&e_j {e \grave{}}_{k} +e_k {e \grave{}}_{j} = 0\\
\end{array}
\]

\[
e_i, {e \grave{}}_{i}  \mbox{ commute with all elements $x_i,{x \grave{}}_{i},d x_i,d {x \grave{}}_{i} $}.
\]

The algebra generated by all these elements will be denoted by $\cD\cP_{m,2n}$. Except for the Clifford generators, this algebra is similar to the one studied e.g. in the classical works by Berezin (\cite{MR0208930,MR732126}), Leites (\cite{MR565567}) and Kostant (\cite{MR0580292}). For a more unified approach see e.g. \cite{MR1175751}.
In this paper we will restrict ourselves to this polynomial algebra, but off course all kinds of completions can also be studied.

The introduction of the Clifford algebra elements will allow us to form combinations of generators where there is no real difference between commuting and anti-commuting co-ordinates and this helps to `hide away' the $\mZ_2$-grading. Moreover we will construct a set of operators, acting on $\cD\cP_{m,2n}$, which together generate an $\mathfrak{osp}(1|2)$ Lie superalgebra (see subsection \ref{diracoperator}). 

Note that we only consider superspaces with an even number of anti-commuting co-ordinates. This is necessary for the symplectic Clifford algebra (or Weyl algebra) generated by the ${e \grave{}}_i$ to be non-degenerate.

\subsection{The basic differential operators and contractors}

The basic endomorphisms acting on $\cD\cP_{m,2n}$ are the operators $\pj,\pjb,\pjcont$ and $\pjbcont$. They are defined as follows. First of all we put $\pj(1)=\pjb(1)=\pjcont(1)=\pjbcont(1)=0$. Next for an arbitrary element  $F$ of $\cD\cP_{m,2n}$, we have the following calculus rules:

\vspace{2mm}
The (classical) derivative $\pj$ satisfies
\vspace{-2mm}
\[
\begin{array}{l} 
\pj(x_k  F) = \delta_{jk} F + x_k  \pj F\\
\pj({x \grave{}}_{k} F) = {x \grave{}}_{k} \pj F\\
\pj(dx_k F) = d x_k \pj F\\
\pj(d {x \grave{}}_{k} F) = d {x \grave{}}_{k} \pj F.\\
\end{array}
\]

The partial derivative $\pjb$ with respect to an anti-commuting variable satisfies
\vspace{-2mm}
\[
\begin{array}{l} 
\pjb(x_k  F) = x_k  \pjb F\\
\pjb({x \grave{}}_{k} F) = \delta_{jk} F - {x \grave{}}_{k} \pjb F\\
\pjb(dx_k F) = d x_k \pjb F\\
\pjb(d {x \grave{}}_{k} F) = -d {x \grave{}}_{k} \pjb F.\\
\end{array}
\]

The contraction $\pjcont$with a commuting variable satisfies
\vspace{-2mm}
\[
\begin{array}{l} 
\pjcont(x_k  F) = x_k  \pjcont F\\
\pjcont( {x \grave{}}_{k} F) = {x \grave{}}_{k} \pjcont F\\
\pjcont(dx_k F) = \delta_{jk} F - d x_k  \pjcont F\\
\pjcont(d {x \grave{}}_{k} F) = -d {x \grave{}}_{k} \pjcont F.\\
\end{array}
\]

Finally, the contraction $\pjbcont$ with an anti-commuting variable satisfies
\vspace{-2mm}
\[
\begin{array}{l} 
\pjbcont(x_k  F) = x_k  \pjbcont F\\
\pjbcont( {x \grave{}}_{k} F) = - {x \grave{}}_{k} \pjbcont F\\
\pjbcont(dx_k  F) = -d x_k \pjbcont F\\
\pjbcont(d {x \grave{}}_{k} F) = \delta_{jk} F + d {x \grave{}}_{k} \pjbcont F.\\
\end{array}
\]

Moreover all operators commute with $e_i$ and ${e \grave{}}_i$. This suffices to compute the action of $\pj,\pjb,\pjcont$ and $\pjbcont$ on any element of $\cD\cP_{m,2n}$.

From the previous relations we immediately obtain the following commutation relations for the operators $\pj,\pjb,\pjcont$ and $\pjbcont$, which will be needed later on:
\vspace{-2mm}
\[
\begin{array}{l} 
\pj \mbox{ commutes with $\pk,\pjb,\pjcont, \pjbcont$}\\
\pjb \pkb = - \pkb \pjb\\
\pjb \pkcont =  \pkcont \pjb\\
\pjb \pkbcont = - \pkbcont \pjb\\
\pjcont \pkcont = - \pkcont \pjcont\\
\pjcont \pkbcont = - \pkbcont \pjcont\\
\pjbcont \pkbcont =  \pkbcont \pjbcont.\\
\end{array}
\]

Note that similar definitions can be given for the action on the right.

\subsection{The Dirac operator and some other important operators}
\label{diracoperator}

Let us first give the definition of the exterior derivative in co-ordinates:

\[
d = \sum_{i=1}^m dx_i \pI  + \sum_{i=1}^{2n} d {x \grave{}}_{i} \pIb.
\label{chainrule}
\]
The square of $d$ is zero as expected.

Now we can make combinations of Clifford algebra elements with co-ordinates, differentials and operators. First we introduce super-vector variables $x$ and super-vector differentials $dx$ by

\begin{eqnarray*}
x &=& \sum_{j=1}^m x_j e_j + \sum_{j=1}^{2n} {x \grave{}}_j {e \grave{}}_j\\
dx &=& \sum_{j=1}^m d x_j e_j + \sum_{j=1}^{2n} d{x \grave{}}_j {e \grave{}}_j.
\end{eqnarray*}

In a similar way we introduce a super-Dirac operator and a super-Dirac contractor by

\begin{eqnarray*}
\px &=& -\sum_{j=1}^m e_j \pj +2 \sum_{j=1}^{n} ( {e \grave{}}_{2j} \partial_{{x \grave{}}_{2j-1}} - {e \grave{}}_{2j-1} \partial_{{x \grave{}}_{2j}}  )\\
\pxcont &=&-\sum_{j=1}^m e_j \pjcont +2 \sum_{j=1}^{n} ( {e \grave{}}_{2j} \partial_{{x \grave{}}_{2j-1}}\rfloor - {e \grave{}}_{2j-1} \partial_{{x \grave{}}_{2j}}\rfloor ).
\end{eqnarray*}

\noindent
Now it is easy to calculate

\[
\px x = m-2n = M =\pxcont dx.
\]

This numerical parameter $M$ will be called the super-dimension. It pops up in several formulae and acts as a kind of global dimension. Morover it can be given an interesting physical interpretation, namely twice the ground-level energy of a harmonic oscillator containing $m$ bosonic and $2n$ fermionic degrees of freedom (see our upcoming paper \cite{DBS3}).

Also the squares of $x$ and $\px$ are important; we find

\begin{eqnarray*}
x^2 &=&\sum_{j=1}^n {x\grave{}}_{2j-1} {x\grave{}}_{2j}  -  \sum_{j=1}^m x_j^2\\
\px^2 &=& 4 \sum_{j=1}^n \partial_{{x \grave{}}_{2j-1}} \partial_{{x \grave{}}_{2j}} -\sum_{j=1}^{m} \pj^2.
\end{eqnarray*}
$x^2$ is a generalization of the norm squared of a vector in Euclidean space, $\px^2=\Delta$ will be called the super-Laplace operator.

It is possible to generalize the Euler-operator to superspace as follows:

\[
\mE = \sum_{j=1}^m x_j \pj+\sum_{j=1}^{2n} {x \grave{}}_{j} \pjb.
\]
One can also introduce the Euler-contractor as

\[
\begin{array}{l}
\mE \rfloor=\sum_{j=1}^m d x_j \pjcont + \sum_{j=1}^{2n} d {x \grave{}}_{j} \pjbcont.
\end{array}
\]
The Euler-operator measures the degree of a polynomial, whereas the Euler-contractor measures the degree of a form. This leads to the following decomposition of $\cD\cP_{m,2n}$:

\[
\cD\cP_{m,2n} = \bigoplus_{i,j=0}^{\infty} \Omega_i^j
\]
with

\[
\Omega_l^k=\left\{ \omega \in \cD\cP_{m,2n} \; | \; \mE \omega=k \omega \mbox{ and } \mE \rfloor \omega=l \omega \right\}.
\]
We will refer to an element of $\Omega_l^k$ as an $l$-form, homogeneous of degree $k$.

Furthermore we have a.o. the following relations (see \cite{DBS1})

\begin{eqnarray*}
\{x, \px\} &=& 2 \mE +M\\
\left[\Delta, x^2\right] &=& 2( 2\mE+M).
\end{eqnarray*}
Note that the second relation is familiar from the basic theory of harmonic analysis (see e.g. \cite{MR1151617}).
From these relations it follows that, putting $G_0 = <\Delta, x^2, 2\mE+M>$ and $G_1 = <x,\px>$, $G=G_0 \oplus G_1$ equipped with the standard graded commutator, is a Lie superalgebra of type $\mathfrak{osp}(1|2)$ (see e.g. \cite{MR0486011}).

Finally we introduce the super Gamma operator $\Gamma$ and the Laplace-Beltrami operator $\Delta_{LB}$ as

\begin{eqnarray*}
\Gamma &=& x \px - \mE\\
\Delta_{LB} &=& (M-2-\Gamma)\Gamma.
\end{eqnarray*}
One easily calculates that

\[
\Delta_{LB} =x^2 \Delta - \mE(M-2 + \mE)
\]
which is similar to the classical expression with $m$ replaced by $M$. 

We end this subsection by another important property which will be needed in the sequel, the so-called first basic identity (see \cite{MR1012510} and \cite{DBS1})
\[
\px = d \pxcont+\pxcont d.
\]

\subsection{The exterior derivative revisited}

We end this section by another characterization of the exterior derivative. This lemma will provide us with the necessary ideas for the following sections.

\begin{lemma}
The operator $d$ is the anti-commutator of the super-vector differential and the super-Dirac operator, i.e.

\[
d = \frac{1}{2} \{ dx, \px\}.
\]
\label{radialdefextdiff}
\end{lemma}

\begin{proof}
We expand $\{ dx, \px\}$ as

\begin{eqnarray*}
\{ dx, \px\} &=& \{d\ux+d\uxb, \upxb-\upx\}\\
&=& \{ d\ux, \upxb\} - \{ d\ux, \upx\}+\{ d\uxb, \upxb\}-\{ d\uxb, \upx\}.
\end{eqnarray*}
The first and the last term equal zero. The second term yields the bosonic part of the exterior derivative

\begin{eqnarray*}
\{ d\ux, \upx\}&=& \left\{ \sum_i dx_i e_i, \sum_j \pj e_j  \right\}\\
&=& \sum_{i,j}dx_i \pj (e_i e_j + e_j e_i)\\
&=& - 2 \sum_{i}dx_i \pI.
\end{eqnarray*}
and the third term gives us the fermionic part

\begin{eqnarray*}
\{ d\uxb, \upxb\}&=& \left\{ \sum_i ( {e \grave{}}_{2i-1} d{x \grave{}}_{2i-1}+{e \grave{}}_{2i} d{x \grave{}}_{2i}), 2 \sum_{j} ( \fermdirac )  \right\}\\
&=& \left\{ \sum_i  {e \grave{}}_{2i-1} d{x \grave{}}_{2i-1} , 2 \sum_{j}  {e \grave{}}_{2j} \partial_{{x \grave{}}_{2j-1}}  \right\}- \left\{ \sum_i {e \grave{}}_{2i} d{x \grave{}}_{2i}  , 2 \sum_{j}  {e \grave{}}_{2j-1} \partial_{{x \grave{}}_{2j}}    \right\}\\ 
&=& 2\sum_{i,j} d{x \grave{}}_{2i-1} \partial_{{x \grave{}}_{2j-1}} ({e \grave{}}_{2i-1} {e \grave{}}_{2j} - {e \grave{}}_{2j} {e \grave{}}_{2i-1})\\
&&-2\sum_{i,j} d{x \grave{}}_{2i} \partial_{{x \grave{}}_{2j}} ({e \grave{}}_{2i} {e \grave{}}_{2j-1} - {e \grave{}}_{2j-1} {e \grave{}}_{2i})\\
&=& 2 \sum_{i} ( d{x \grave{}}_{2i-1} \partial_{{x \grave{}}_{2i-1}} + d{x \grave{}}_{2i} \partial_{{x \grave{}}_{2i}}).
\end{eqnarray*}
By adding all terms the proof is completed.
\end{proof}

\section{The Hodge coderivative}

In Lemma \ref{radialdefextdiff} we have proven that the exterior derivative can be written as the anti-commutator of $dx$ and $\px$. Therefore it would also be interesting to study other anti-commutators between the canonical objects $x,\px, \pxcont$ and $dx$.

Indeed, we can for example introduce a super-Hodge coderivative $d^*$ by

\begin{equation}
d^* = \frac{1}{2} \{ \pxcont, \px\}.
\label{hodgedef}
\end{equation}
After some calculations this leads to the following co-ordinate expression

\[
d^* = -\sum_{j=1}^m \pjcont \pj +2 \sum_{i=1}^n (\partial_{{x \grave{}}_{2i-1}} \rfloor \partial_{{x \grave{}}_{2i}} -\partial_{{x \grave{}}_{2i}} \rfloor \partial_{{x \grave{}}_{2i-1}}).
\]

We denote by $d^*_b$ and $d^*_f$ the respective bosonic and fermionic part of the operator $d^*$:

\[
\begin{array}{l}
d^*_b = -\sum_{j=1}^m \pjcont \pj \\
\vspace{-2mm}\\
d^*_f = 2 \sum_{i=1}^n (\partial_{{x \grave{}}_{2i-1}} \rfloor \partial_{{x \grave{}}_{2i}} -\partial_{{x \grave{}}_{2i}} \rfloor \partial_{{x \grave{}}_{2i-1}}).
\end{array}
\]

We state the most important properties of this operator in the following theorem. They clearly justify the name Hodge coderivative given to $d^*$.

\begin{theorem}
One has
\begin{enumerate}
\item $(d^*)^2=0$
\item $d d^* + d^* d = \Delta.$
\end{enumerate}
\end{theorem}

\begin{proof}

We calculate

\[
d^{*2}= d^{*2}_b +d^*_b d^*_f + d^*_f d^*_b +d^{*2}_f.
\]
The first term is zero (this is the classical property of the Hodge operator). The second term cancels the third.
Finally, the calculation of $d^{*2}_f$ can be done as follows:

\begin{eqnarray*}
d^{*2}_f &=& 4  \sum_{i,j=1}^n (\partial_{{x \grave{}}_{2i-1}} \rfloor \partial_{{x \grave{}}_{2i}} -\partial_{{x \grave{}}_{2i}} \rfloor \partial_{{x \grave{}}_{2i-1}})  (\partial_{{x \grave{}}_{2j-1}} \rfloor \partial_{{x \grave{}}_{2j}} -\partial_{{x \grave{}}_{2j}} \rfloor \partial_{{x \grave{}}_{2j-1}})\\
&=&-4  \sum_{i,j=1}^n (\partial_{{x \grave{}}_{2j-1}} \rfloor \partial_{{x \grave{}}_{2j}} -\partial_{{x \grave{}}_{2j}} \rfloor \partial_{{x \grave{}}_{2j-1}}) (\partial_{{x \grave{}}_{2i-1}} \rfloor \partial_{{x \grave{}}_{2i}} -\partial_{{x \grave{}}_{2i}} \rfloor \partial_{{x \grave{}}_{2i-1}})\\
&=&0.
\end{eqnarray*}
which completes the proof of the first part.

Now we prove that $d d^* + d^* d = \Delta$. We use the definition and the first basic identity to obtain

\begin{eqnarray*}
2 d d^* &=& d \pxcont \px + d \px \pxcont\\
&=& \px^2 - \pxcont d \px + d \px \pxcont\\
2  d^* d &=& \pxcont \px d + \px \pxcont d\\
&=& \pxcont \px d + \px^2 - \px d \pxcont.
\end{eqnarray*}
Now it suffices to note that $[\px,d] = 0$ to complete the proof.
\end{proof}

\begin{remark}
Although we are able to define the operator $d^*$, we cannot define a super analogue of the Hodge star map. Classically this is a map from the space of $k$-forms to the space of $(n-k)$-forms, with $n$ the dimension of the manifold under consideration. This is no longer possible since the de Rham complex is now infinite.
\end{remark}

\section{The Poincar\'e and the Dual Poincar\'e Lemma}

From the previous sections it readily follows that $d,d^*$ and $\Delta$ are maps between the following spaces

\[
\begin{array}{ll}
d: & \Omega_k^l \rightarrow \Omega_{k+1}^{l-1}\\
d^*: & \Omega_k^l \rightarrow \Omega_{k-1}^{l-1}\\
\Delta: & \Omega_k^l \rightarrow \Omega_{k}^{l-2}.\\
\end{array}
\]

We can visualize these operators (see figures \ref{triangle1} and \ref{triangle2}).
In these diagrams the operator $d$ is a horizontal arrow, $d^*$ is a vertical arrow and $\Delta$ is a diagonal arrow (to simplify the figures we have only included one). Also note that these diagrams are not commutative. 

It is also worthwhile to take a look at the bosonic and fermionic limit. In the bosonic limit (e.g. $m=4,n=0$) figure \ref{triangle1} reduces to figure \ref{boslim}, because there is a bound on the subindex of $\Omega_k^l$. In the fermionic limit (e.g. $m=0,n=1$) figure \ref{triangle1} reduces to figure \ref{ferlim}, since in this case there is a bound on the upper index of $\Omega_k^l$. 

We can now formulate super versions of the classical lemmata of Poincar\'e in $\cD\cP_{m,2n}$. The inspiration for the proof stems from \cite{MR2186030} and from the abstract notions on Clifford analysis in \cite{MR1472163}.

First we define operators $S$ and $T$ by

\begin{eqnarray*}
S &=&\frac{1}{2} \{x,\pxcont\} \\
T &=&\frac{1}{2} \{x , dx\}.\\
\end{eqnarray*}

They have the following co-ordinate expression, as can be seen after some calculations

\begin{eqnarray*}
S &=& \sum_{j=1}^m x_j \pjcont  +\sum_{j=1}^{2n} {x \grave{}}_{j} \pjbcont\\
T &=& - \sum_{j=1}^m x_j d x_j  +\frac{1}{2} \sum_{j=1}^{n} ({x \grave{}}_{2j-1} d{x \grave{}}_{2j}-{x \grave{}}_{2j} d{x \grave{}}_{2j-1})
\end{eqnarray*}
and they both square to zero. Moreover, we have that

\[
\begin{array}{ll}
S: & \Omega_k^l \rightarrow \Omega_{k-1}^{l+1}\\
T: & \Omega_k^l \rightarrow \Omega_{k+1}^{l+1},\\
\end{array}
\]
so $S$ is a suitable candidate for a homotopy operator for $d$, and so is $T$ for $d^*$. Indeed we have the following technical lemma:

\begin{lemma}
One has
\begin{eqnarray*}
d S + S d &=& \mE+\mE \rfloor\\
d^* T + T d^* &=& M + \mE -\mE \rfloor
\end{eqnarray*}
where $M=m-2n$.
\label{homrels}
\end{lemma}

\begin{proof}
We only give the proof for the second identity.
We calculate $d^* T$

\[
d^* T = d^*_b T_b + d^*_b T_f + d^*_f T_b + d^*_f T_f
\]
with

\begin{eqnarray*}
T_b&=& - \sum_{j=1}^m x_j d x_j\\
T_f&=& \frac{1}{2} \sum_{j=1}^{n} ({x \grave{}}_{2j-1} d{x \grave{}}_{2j}-{x \grave{}}_{2j} d{x \grave{}}_{2j-1}).
\end{eqnarray*}
We calculate the first term

\begin{eqnarray*}
d^*_b T_b &=& \sum_{i,j} \pIcont \pI x_j d x_j\\
&=& \sum_{i,j} \pIcont  d x_j (\delta_{ij} + x_j\pI )\\
&=& \sum_{i,j} (\delta_{ij}- dx_j \pIcont) (\delta_{ij} + x_j\pI )\\
&=& m - \sum_i dx_i \pIcont + \sum_i x_i \pI - \sum_{i,j} x_j d x_j \pIcont \pI \\
&=&m- \mE_b \rfloor + \mE_b - T_b d_b^*.
\end{eqnarray*}
The second and third term give
\begin{eqnarray*}
d^*_b T_f &=&-T_f d_b^*\\
d^*_f T_b &=& - T_b d^*_f.
\end{eqnarray*}
Finally we calculate the last term

\begin{eqnarray*}
d^*_f T_f &=& \sum_{i,j} (\partial_{{x \grave{}}_{2i-1}} \rfloor \partial_{{x \grave{}}_{2i}}-\partial_{{x \grave{}}_{2i}} \rfloor \partial_{{x \grave{}}_{2i-1}})( {x \grave{}}_{2j-1} d {x \grave{}}_{2j} -{x \grave{}}_{2j} d {x \grave{}}_{2j-1} )\\
&=& \sum_{i,j} \partial_{{x \grave{}}_{2i-1}} \rfloor \partial_{{x \grave{}}_{2i}}  {x \grave{}}_{2j-1} d {x \grave{}}_{2j} \quad (a)\\
&&- \sum_{i,j} \partial_{{x \grave{}}_{2i-1}} \rfloor \partial_{{x \grave{}}_{2i}}  {x \grave{}}_{2j} d {x \grave{}}_{2j-1} \quad (b)\\
&&- \sum_{i,j} \partial_{{x \grave{}}_{2i}} \rfloor \partial_{{x \grave{}}_{2i-1}}  {x \grave{}}_{2j-1} d {x \grave{}}_{2j} \quad (c)\\
&&+ \sum_{i,j} \partial_{{x \grave{}}_{2i}} \rfloor \partial_{{x \grave{}}_{2i-1}}  {x \grave{}}_{2j} d {x \grave{}}_{2j-1} \quad (d).\\
\end{eqnarray*}
The terms $(a)$ and $(d)$ yield

\begin{eqnarray*}
(a)&=&-\sum_{i,j}  {x \grave{}}_{2j-1} d {x \grave{}}_{2j} \partial_{{x \grave{}}_{2i-1}} \rfloor \partial_{{x \grave{}}_{2i}} \\
(d)&=& -\sum_{i,j}  {x \grave{}}_{2j} d {x \grave{}}_{2j-1} \partial_{{x \grave{}}_{2i}} \rfloor \partial_{{x \grave{}}_{2i-1}}.\\
\end{eqnarray*}
We calculate $(b)$ as

\begin{eqnarray*}
(b) &=&  - \sum_{i,j} \partial_{{x \grave{}}_{2i-1}} \rfloor d {x \grave{}}_{2j-1} (\delta_{ij}-   {x \grave{}}_{2j} \partial_{{x \grave{}}_{2i}} )\\
&=&-  \sum_{i,j} (\delta_{ij} + d {x \grave{}}_{2j-1} \partial_{{x \grave{}}_{2i-1}}  \rfloor) (\delta_{ij}-   {x \grave{}}_{2j}  \partial_{{x \grave{}}_{2i}} )\\
&=& - n - \sum_i d {x \grave{}}_{2i-1} \partial_{{x \grave{}}_{2i-1}} \rfloor +\sum_i {x \grave{}}_{2i} \partial_{{x \grave{}}_{2i}} + \sum_{i,j}  {x \grave{}}_{2j} d {x \grave{}}_{2j-1} \partial_{{x \grave{}}_{2i-1}} \rfloor \partial_{{x \grave{}}_{2i}}.
\end{eqnarray*}
Similarly, (c) yields

\begin{eqnarray*}
(c) &=&- n - \sum_i d {x \grave{}}_{2i} \partial_{{x \grave{}}_{2i}} \rfloor +\sum_i {x \grave{}}_{2i-1} \partial_{{x \grave{}}_{2i-1}} + \sum_{i,j}  {x \grave{}}_{2j-1} d {x \grave{}}_{2j} \partial_{{x \grave{}}_{2i}} \rfloor \partial_{{x \grave{}}_{2i-1}}.
\end{eqnarray*}
Collecting all terms we find that

\[
d_f^* T_f = -2n - \mE_f \rfloor + \mE_f - T_f d_f^*.
\]
Adding all previous terms concludes the proof of the lemma.
\end{proof}

Note once again the natural appearance of the super-dimension $M$ in the second formula of the previous lemma.

Using these results, we can immediately prove the lemma of Poincar\'e (which can e.g. also be found in \cite{MR1175751} and \cite{MR0580292} without using Clifford numbers):

\begin{lemma}[Poincar\'e]
In figures (\ref{triangle1}) and (\ref{triangle2}) every row is exact, except for the first row in figure (\ref{triangle1}),
\[
\xymatrix{
0\ar[r] &\Omega^{0}_0\ar[r]& 0.\\
}
\]
\end{lemma}

\begin{proof}
Take $\omega \in \Omega^k_{l}$ with $d \omega = 0$. Then, by the previous lemma,

\[
(k+l) \omega = S d \omega  +d S \omega.\\
\]
or

\[
\omega = \frac{1}{k+l} d (S \omega),
\]
which works for all $\Omega_{l}^{k}$, except for $\Omega_0^0$.
\end{proof}

We can also state the dual of the previous lemma, which cannot be found in the literature so far. Note that the proof is surprisingly more difficult than classically.

\begin{lemma}[Dual Poincar\'e]
In figures (\ref{triangle1}) and (\ref{triangle2}) every column is exact, except at $\Omega_{m}^{2n}$, where the homology-space is one-dimensional and a representative is given by 
\[
{x \grave{}}_{1} {x \grave{}}_{2} \ldots {x \grave{}}_{2n} dx_1 dx_2 \ldots dx_m.
\]
\end{lemma}

\begin{proof}

We have to consider two different cases. Take $\omega \in \Omega^k_l$ with $d^* \omega=0$. 
First suppose that $m-2n + k-l \neq 0$. Using lemma \ref{homrels} we have that

\begin{eqnarray*}
(d^* T + T d^*) \omega & =&  (m-2n + \mE -\mE \rfloor) \omega\\
&=& (m-2n + k -l) \omega
\end{eqnarray*}
or

\[
\omega = \frac{1}{k+m-2n-l} d^* (T \omega).
\]

So it remains to examine the case where $k+m-2n-l = 0$. This corresponds with exactly one column in either figure \ref{triangle1} or \ref{triangle2}, depending on the value of $m$. First note that the theorem is true in the bosonic and fermionic limit (in both cases the respective column contains only one non-zero space).

Next suppose that $m,n \neq 0$. We use an iterative argument. Take $\omega \in \Omega^k_l$ with $k+m-2n-l = 0$ and $d^* \omega=0$. Then we can write $\omega$ as

\[
\omega = \sum_{i=0}^k x_1^i \omega_0^i + dx_1 \sum_{i=0}^k x_1^i \omega_1^i
\]
with $x_1,dx_1 \not \in \omega_0^i, \omega_1^i$.
Expressing the fact that $d^* \omega=0$ leads to the following equations:

\[
\begin{array}{ll}
d^* \omega_0^0 = \omega_1^1&d^*\omega_1^1=0\\
d^* \omega_0^1 = 2 \omega_1^2&d^*\omega_1^2=0\\
\vdots&\vdots\\
d^* \omega_0^i = (i+1) \omega_1^{i+1}&d^*\omega_1^{i+1}=0\\
\vdots&\vdots\\
d^* \omega_0^{k-1} = k \omega_1^{k}&d^*\omega_1^k=0\\
d^* \omega_0^{k} =0&\\
d^* \omega_1^{0} =0.&\\
\end{array}
\]

As $\omega_0^{k} \in \Omega_l^0(m-1,2n)$ (we replace $m$ by $m-1$ since $x_1$ does not appear in $\omega_0^{k}$) we have that $m-1-2n-l\neq0$, so by the first part of the proof there exists a $\gamma$ with $d^* \gamma =\omega_0^{k}$. Putting

\[
\alpha = - \sum_{i=1}^k dx_1 \frac{x_1^i}{i} \omega_0^{i-1} + x_1^k \gamma
\]
we obtain that

\[
d^* \alpha = \omega - dx_1 \omega_1^0.
\]

So we still have to examine the term $dx_1 \omega_1^0$. We have that $\omega_1^{0} \in \Omega_{l-1}^k(m-1,2n)$ with $d^*\omega_1^{0} =0 $ and $k+m-1-2n-(l-1) = 0$.
Now we can repeat the previous argument by splitting off the next commuting variable $x_2, dx_2$ in $\omega_1^{0}$.
This procedure can be repeated. We distinguish three different cases:

\begin{itemize}
\item The case where $l<m$

After $l$ steps the following term remains
\[
dx_1 \ldots dx_l \alpha,
\]
where $\alpha \in \Omega_0^k(m-l,2n)$ can be written as

\[
\alpha = \sum_{i=0}^k x_{l+1}^i \alpha_i, \quad x_{l+1} \not \in \alpha_i, \alpha_i \in \Omega_{0}^{k-i}.
\]

So if we put 

\[
\beta = - \sum_{i=0}^k dx_{l+1} \frac{x_{l+1}^{i+1}}{i+1} \alpha_i
\]

then

\[
d^* \beta = \alpha
\]

and 

\[
d^*( (-1)^l dx_1 \ldots dx_l \beta )= dx_1 \ldots dx_l \alpha.
\]

\item The case where $l>m$

In this case the remaining term vanishes. Indeed, if this is not the case then, after $m$ steps, we would end up with a term of the form

\[
dx_1 \ldots dx_m \alpha,\quad \alpha \in \Omega_{l-m}^k(0,2n).
\]

implying that $k > 2n$ which is not possible.

\item The case where $l=m$

This implies that $k=2n$. After $m$ steps the following term remains
\[
dx_1 \ldots dx_m \alpha,\quad \alpha \in \Omega_0^{2n}(0,2n).
\]

so $\alpha ={x \grave{}}_{1} {x \grave{}}_{2} \ldots {x \grave{}}_{2n} $ and we have to examine

\[
\delta = {x \grave{}}_{1} {x \grave{}}_{2} \ldots {x \grave{}}_{2n}  dx_1 \ldots dx_m \in \Omega^{2n}_{m}.
\]

It is now easy to see that $d^* \delta = 0$ and that there does not exist an $\epsilon$ such that $d^* \epsilon = \delta$. This completes the proof of the lemma.
\end{itemize}
\end{proof}

\begin{remark} 

(1) The homotopy operators $S$ and $T$ used in the proofs of the above theorems are defined in a very natural way using Clifford analysis, by considering the right anti-commutators. The only problem that occurs in the last proof is that $d^* T+ T d^*$ can become zero.

(2) The results of the last theorem can be given a nice graphical interpretation. If we consider a superspace of dimension $(m,2n)$, then we can consider both the bosonic $(m,0)$ and fermionic $(0,2n)$ limit. If we put both diagrams on top of each other, we notice that the crossing point is the space $\Omega_{m}^{2n}$, which is exactly the space where the homology does not vanish.
\end{remark}

\section{Connection with the Berezin integral: a super volume-form}

In classical analysis on $\mR^m$, the volume-form $dx_1 \wedge dx_2\wedge \ldots \wedge dx_m$ is the unique form that is co-closed but not co-exact (in a simply-connected domain). In superspace we find similarly the form ${x \grave{}}_{1} {x \grave{}}_{2} \ldots {x \grave{}}_{2n}  dx_1 \ldots dx_m$. This clearly establishes a connection with the Berezin-integral (see \cite{MR732126}). There the commuting variables are integrated in the classical way, but the anti-commuting ones all give zero except for the term in ${x \grave{}}_{1} {x \grave{}}_{2} \ldots {x \grave{}}_{2n}$.

To clarify this, consider a super-function $f$ in the sense of Berezin defined in a domain $\Omega \in \mR^m$. This function $f$ can be written uniquely in the following form

\begin{equation}
f(x, {x \grave{}}) = \sum_{\nu=(\nu_1,\ldots,\nu_{2n})} f_{\nu}(x) {x \grave{}}_{1}^{\nu_1} \ldots {x \grave{}}_{2n}^{\nu_{2n}}
\label{expansion}
\end{equation}
where $\nu_i = 0$ or $1$ and $f_{\nu}(x)$ is a smooth function of the (real) co-ordinates $(x_1,\ldots,x_m)$. The Berezin-integral of this function is defined as follows

\begin{eqnarray*}
\int_B f(x, {x \grave{}}) &=& \int_{B}  f_{(1,1,\ldots,1)}(x) {x \grave{}}_{1} \ldots {x \grave{}}_{2n}\\
&=& \int_{\Omega} f_{(1,1,\ldots,1)}(x) dx_1 dx_2 \ldots dx_m
\end{eqnarray*}
which is sometimes multiplied by a constant factor. We see that only the term of $f$ in ${x \grave{}}_{1} \ldots {x \grave{}}_{2n}$ gives a contribution to the integral.

This is completely in correspondance with the results from the previous section. Indeed, we have obtained that $\omega = {x \grave{}}_{1} {x \grave{}}_{2} \ldots {x \grave{}}_{2n}  dx_1 \ldots dx_m$is the unique form that is co-closed and not co-exact. So it is a suitable candidate for a volume-form. Moreover if we multiply $\omega$ with a function $f$ of the form (\ref{expansion}), we find that the only types of objects that one can integrate are of the following form $g(x) \omega = g(x){x \grave{}}_{1} {x \grave{}}_{2} \ldots {x \grave{}}_{2n}  dx_1 \ldots dx_m$ with $g(x)$ a smooth function of the real co-ordinates $x_i$. Also the integration-recipe is now rather obvious, if one takes into account that in the Berezin-Leites approach to superspace the $x_i$-variables take real values and the ${x \grave{}}_i$-variables are only abstract objects.

\section{Conclusions and further research}

In this paper we have given a framework for constructing Clifford analysis on superspace. More specifically, the introduction of orthogonal and symplectic Clifford generators allows one to construct a super-Dirac operator, super vector-variables, and the likes. This also leads in a natural way to the introduction of the so-called superdimension.

Furthermore we have considered anti-commutators between these canonical objects, which has enabled us to introduce a Hodge operator on superspace and to construct the suitable homotopy-operators. We have expanded the classical polynomial de Rham complex and obtained some enlightening graphical interpretations. Although the proof of the classical Poincar\'e lemma stays virtually identical, the proof of its dual is a lot more complicated. However, this gives some interesting information concerning the Berezin-integral.

In subsequent papers we will further develop this approach to superspace. We envisage an extension of the Fischer-decomposition to the super case, a detailed study of super spherical monogenics and harmonics, a construction of new types of special functions and possible physical applications.

\begin{figure}[h]
\[
\xymatrix@=15pt{&&&&0&&&&\\
&&&0\ar[r] &\Omega^{0}_0\ar[r]\ar[u] & 0&&&\\
&&0  \ar[r]  &\Omega^{2}_0\ar[r] \ar[u] & \Omega^{1}_1\ar[r] \ar[u]&  \Omega^{0}_2\ar[r] \ar[u] &0&&\\
&0  \ar[r] &\Omega^{4}_0\ar[r] \ar[u] & \Omega^{3}_1\ar[r] \ar[u]& \Omega^{2}_2\ar[r] \ar[u] &\Omega^{1}_3\ar[r] \ar[u] &\Omega^{0}_4\ar[r] \ar[u]&0& \\
0  \ar[r] &\Omega^{6}_0\ar[r] \ar[u] & \Omega^{5}_1\ar[r] \ar[u]& \Omega^{4}_2\ar[r] \ar[u] \ar[ur] &\Omega^{3}_3\ar[r] \ar[u] &\Omega^{2}_4\ar[r] \ar[u] &\Omega^{1}_5\ar[r] \ar[u]&\Omega^{0}_6\ar[r] \ar[u]&0 \\
&\ar[u]&\ar[u]&\ar[u]&\ar[u]&\ar[u]&\ar[u]&\ar[u]&}
\]
\caption{Graphical interpretation: spaces of polynomial-valued differential forms.}
\label{triangle1}
\end{figure}

\begin{figure}[h]
\[
\xymatrix@=15pt{&&&&0&0&&&&\\
&&&0\ar[r] &\Omega^{1}_0\ar[r]\ar[u]&\Omega^{0}_1\ar[r]\ar[u] & 0&&&\\
&&0  \ar[r]  &\Omega^{3}_0\ar[r] \ar[u] & \Omega^{2}_1\ar[r] \ar[u]&  \Omega^{1}_2\ar[r] \ar[u] &\Omega^{0}_3\ar[r]\ar[u]&0&&\\
&0  \ar[r] &\Omega^{5}_0\ar[r] \ar[u] & \Omega^{4}_1\ar[r] \ar[u]& \Omega^{3}_2\ar[r] \ar[u] &\Omega^{2}_3\ar[r] \ar[u] &\Omega^{1}_4\ar[r] \ar[u]&\Omega^{0}_5\ar[r]\ar[u]&0& \\
0  \ar[r] &\Omega^{7}_0\ar[r] \ar[u] & \Omega^{6}_1\ar[r] \ar[u]& \Omega^{5}_2\ar[r] \ar[u] &\Omega^{4}_3\ar[r] \ar[u] &\Omega^{3}_4\ar[r] \ar[u]&\Omega^{2}_5\ar[r] \ar[u]&\Omega^{1}_6\ar[r] \ar[u]&\Omega^{0}_7\ar[r]\ar[u]&0 \\
&\ar[u]&\ar[u]&\ar[u]&\ar[u]&\ar[u]&\ar[u]&\ar[u]&\ar[u]&}
\]
\caption{Graphical interpretation: spaces of polynomial-valued differential forms.}
\label{triangle2}
\end{figure}

\begin{figure}
\[
\xymatrix@=15pt{&&&&&&0&&&&\\
&&&&&0\ar[r] &\Omega^{0}_0\ar[r]\ar[u] & 0&&&\\
&&&&0  \ar[r]  &\Omega^{2}_0\ar[r] \ar[u] & \Omega^{1}_1\ar[r] \ar[u]&  \Omega^{0}_2\ar[r] \ar[u] &0&&\\
&&&0  \ar[r] &\Omega^{4}_0\ar[r] \ar[u] & \Omega^{3}_1\ar[r] \ar[u]& \Omega^{2}_2\ar[r] \ar[u] &\Omega^{1}_3\ar[r] \ar[u] &\Omega^{0}_4\ar[r] \ar[u]&0& \\
&&0  \ar[r] &\Omega^{6}_0\ar[r] \ar[u] & \Omega^{5}_1\ar[r] \ar[u]& \Omega^{4}_2\ar[r] \ar[u] &\Omega^{3}_3\ar[r] \ar[u] &\Omega^{2}_4\ar[r] \ar[u]&0\ar[u]& \\
&0  \ar[r] &\Omega^{8}_0\ar[r] \ar[u] & \Omega^{7}_1\ar[r] \ar[u]& \Omega^{6}_2\ar[r] \ar[u] &\Omega^{5}_3\ar[r] \ar[u] 
&\Omega^{4}_4\ar[r] \ar[u]&0\ar[u]&& \\
&&\ar[u]&\ar[u]&\ar[u]&\ar[u]&\ar[u]&}
\]
\caption{Bosonic limit $m=4, n=0$.}
\label{boslim}
\end{figure}

\begin{figure}
\[
\xymatrix@=15pt{&&0&&&&\\
&0\ar[r] &\Omega^{0}_0\ar[r]\ar[u] & 0&&&\\
0  \ar[r]  &\Omega^{2}_0\ar[r] \ar[u] & \Omega^{1}_1\ar[r] \ar[u]&  \Omega^{0}_2\ar[r] \ar[u] &0&&\\
&0  \ar[r]\ar[u] &\Omega^{2}_2\ar[r] \ar[u] & \Omega^{1}_3\ar[r] \ar[u]& \Omega^{0}_4\ar[r] \ar[u] &0& \\
&&0  \ar[r]\ar[u] &\Omega^{2}_4\ar[r] \ar[u] & \Omega^{1}_5\ar[r] \ar[u]& \Omega^{0}_6\ar[r] \ar[u] &0& \\
&&&0  \ar[r]\ar[u] &\Omega^{2}_6\ar[r] \ar[u] & \Omega^{1}_7\ar[r] \ar[u]& \Omega^{0}_8\ar[r] \ar[u] &0& \\
&&&&\ar[u]&\ar[u]&\ar[u]&}
\]
\caption{Fermionic limit $m=0, n=1$.}
\label{ferlim}
\end{figure}

\end{document}